\documentclass[twoside,a4paper]{article}
\usepackage{amsmath,graphicx,amssymb,fancyhdr,amsthm} 
\newtheorem{thm}{Theorem}[section] 
\newtheorem{cor}[thm]{Corollary}

\theoremstyle{definition} 
  
\theoremstyle{remark}  
  
\def\beq{\begin{eqnarray}}  
\def\eeq{\end{eqnarray}}  
\def\bsp{\begin{split}}  
\def\esp{\end{split}}

\def\Tr{\mathrm{Tr}}  
\def\d{\mathrm{d}}  
\def\diag{\mathrm{diag}}

\newcommand{\mf}[1]{{\mathfrak #1}}

\newcommand{\mbold}[1]{\mbox{\boldmath{\ensuremath{#1}}}}

\begin{document}  
  
\title{\Large\textbf{Higher dimensional bivectors and classification of the Weyl operator}}  
\author{{\large\textbf{Alan Coley$^1$ and Sigbj\o rn Hervik$^2$} }  
 \vspace{0.3cm} \\
$^{1}$Department of Mathematics and Statistics,\\
Dalhousie University,
Halifax, Nova Scotia,\\
Canada B3H 3J5
\vspace{0.3cm}\\
$^{2}$Faculty of Science and Technology,\\  
 University of Stavanger,\\  N-4036 Stavanger, Norway   
\vspace{0.3cm} \\   
\texttt{aac@mathstat.dal.ca, ~sigbjorn.hervik@uis.no} }  
\date{\today}  
\maketitle  
\pagestyle{fancy}  
\fancyhead{} 
\fancyhead[EC]{A. Coley and S. Hervik}  
\fancyhead[EL,OR]{\thepage}  
\fancyhead[OC]{The Weyl operator}  
\fancyfoot{} 
  
\begin{abstract}

We develop the bivector formalism in 
higher  dimensional Lorentzian spacetimes. 
We define the Weyl bivector operator in a manner consistent with 
its boost-weight decomposition.  We then algebraically classify
the Weyl tensor, which gives rise to a refinement  in dimensions
higher than four of the usual
alignment (boost-weight) classification, 
in terms of the irreducible representations of the spins. 
We are consequently able to define a number of new
algebraically special cases. In particular, the classification
in five dimensions is discussed in some detail.
In addition, utilizing the  (refined) algebraic classification, we are able to
prove some interesting results when the Weyl tensor has (additional) symmetries.

\end{abstract} 
\section{Introduction}

Higher dimensional Lorentzian spacetimes are of considerable 
interest in current theoretical physics. Therefore, it is useful 
to have generalisations to higher dimensions of the mathematical 
tools (which have been successfully employed in 4d) to study 
higher dimensional Lorentzian spacetimes. In particular, the 
introduction of the {\em alignment theory} \cite{class} has made 
it possible to algebraically classify any tensor in a Lorentzian 
spacetime of arbitrary dimensions by boost weight, including the 
classification of the Weyl tensor and the Ricci tensor (thus 
generalizing the Petrov and Segre classifications in 4d). In 
addition, using alignment theory the higher dimensional Bianchi 
and Ricci identities have been computed  \cite{PODREF} and a 
higher dimensional generalization of {\em Newman-Penrose 
formalism} has been presented \cite{class}.

It is also of interest to develop other mathematical tools for 
studying higher dimensional Lorentzian spacetimes. Two other types 
of classification can be obtained by introducing {\em bivectors} and 
{\em spinors}. The algebraic classification of the Weyl tensor using 
bivectors or spinors is 
equivalent to the algebraic classification of the Weyl tensor by 
boost weight in 4d (i.e., the Petrov classification \cite{kramer}). However, these 
classifications are different in higher dimensions. In particular, 
the algebraic classification using alignment theory is rather 
coarse, and it may be useful to develop the algebraic 
classification of the Weyl tensor using bivectors or spinors to 
obtain a more refined classification. It is the purpose of this 
paper to develop the bivector formalism in higher dimensions. 
 
Indeed, we are primarily motivated to develop the bivector formalism in 
higher dimensions in order to generalize the various theorems 
presented \cite{inv,kundt,CSI4} to higher dimensions. 
In  \cite{inv} it was shown that in 4d a Lorentzian spacetime 
metric is either $\mathcal{I}$-non-degenerate, and hence locally 
characterized by its scalar polynomial curvature invariants 
constructed from the Riemann tensor and its covariant derivatives, 
or is a degenerate Kundt spacetime. The (higher dimensional) Kundt 
spacetimes admit a kinematic frame in which there exists a null 
vector $\ell$ that is geodesic, expansion-free, shear-free and 
twist-free; the Kundt metric is written in \cite{kundt,Higher}. 
The {\it degenerate} Kundt spacetimes \cite{kundt} are such that 
there exists a common null frame in which the geodesic, 
expansion-free, shear-free and twist-free null vector $\ell$ is 
also the null vector in which all positive boost weight terms of 
the Riemann tensor and its covariant derivatives are zero. 
Therefore, the degenerate Kundt spacetimes are the only spacetimes 
in 4d that are {\em not} $\mathcal{I}$-non-degenerate, and their 
metrics are the only metrics not determined by their curvature 
invariants \cite{inv}. The degenerate Kundt spacetimes were 
classified algebraically by the Riemann tensor and its covariant 
derivatives in \cite{kundt}. Recently,  a number
of  exact higher dimensional solutions  have been studied \cite{other}, 
including a class of exact higher dimensional 
Einstein-Maxwell Kundt spacetimes \cite{POD}.
 
In the proof of the $\mathcal{I}$-non-degenerate theorem in 
\cite{inv}, it was necessary to determine for which Segre types 
for the Ricci tensor the spacetime is 
$\mathcal{I}$-non-degenerate. In each case, it was found that the 
Ricci tensor, considered as a curvature operator, admits a 
timelike eigendirection; therefore, if a spacetime is not 
$\mathcal{I}$-non-degenerate, its Ricci tensor must be of a 
particular Segre type. By analogy, in higher dimensions it is possible to show  that 
if the algebraic type of the Ricci tensor (or any other rank 2 
curvature operator written in `Segre form') is not of one of the 
types $\{(1,1)11...\}$, $\{2111...\}$, $\{(21)11...\}$, $\{(211)11...\}$, 
$\{3111...\}$, etc. (or their degeneracies; for example, 
$\{3(11)1...\}$, etc.), then the spacetime is 
$\mathcal{I}$-non-degenerate. Similar results in terms of the Weyl 
tensor in bivector form in 4d were proven. However, to generalize 
these results to higher dimensions it is necessary to develop the bivector 
formalism for the Weyl tensor in higher dimensions.

Lorentzian spacetimes for which all polynomial scalar invariants 
constructed from the Riemann tensor and its covariant derivatives 
are constant are called CSI spacetimes. All curvature invariants 
of all orders vanish in an $n$-dimensional Lorentzian  VSI 
spacetime \cite{Higher}. If the aligned, repeated, null vector 
$\bf {\ell}$ is also covariantly constant, the spacetime is CCNV. 
Spacetimes which are Kundt-CSI, VSI or CCNV are degenerate-Kundt 
spacetimes. Supersymmetric solutions of supergravity theories have played an 
important role in the development of string theory (see, for 
example, \cite{grover}). Supersymmetric solutions in $M$-theory that 
are not static admit a CCNV \cite{class}. This class includes a 
subset of the Kundt-CSI and the VSI spacetimes as special cases. 
The higher dimensional VSI and CSI degenerate Kundt spacetimes are 
consequently of fundamental importance since they are solutions of 
supergravity or superstring theory, when supported by appropriate 
bosonic fields \cite{CFH}.

In \cite{CSI4} we investigated 4d Lorentzian CSI 
spacetimes and proved that if a 4d spacetime is CSI, then either 
the spacetime is locally homogeneous or the spacetime is a Kundt 
spacetime for which there exists a frame such that the positive 
boost weight components of all curvature tensors vanish and the 
boost weight zero components are all constant. Other 
possible higher dimensional generalizations of the results in
\cite{CSI4} were discussed in \cite{CSI}. Again, the first step is to 
investigate the curvature operators in higher dimensions and to 
classify these (especially the Weyl tensor) for the various 
algebraic types. This necessitates a bivector formalism for the 
Weyl tensor in higher dimensions.

It is also possible to generalize Penrose's spinor calculus 
\cite{PENROSE} of 4d Lorentzian geometry to higher dimensions. 
Recently, partly motivated by the discovery of 
exact  black ring solutions in five (and higher) dimensions 
\cite{HIGHER-D-REVIEW}, a spinor calculus has been explicitly developed in 5d 
\cite{NP5d}. When the spin covariant derivative is 
compatible with the spacetime metric and the symplectic structure 
(in 5d, the spin space is a 4d complex vector space 
endowed with an antisymmetric tensor which plays the role of a 
metric tensor), it can be shown \cite{NP5d} that the spin covariant derivative is unique and
the 5d {\em curvature spinors} (e.g., the Ricci spinor 
and the Weyl spinor) can be defined. 
It is then desirable to generalize to 5d the Newman-Penrose formalism
\cite{NP5d} and the algebraic classification of the 
Weyl spinor \cite{DeSmet,robspin}. In particular, in 5d the  
Weyl tensor can be represented by the Weyl spinor, which 
is equivalent to the Weyl polynomial (which is a homogeneous quartic polynomial 
in 3 variables) \cite{robspin};  the algebraic classification can then be realised by putting 
the Weyl polynomial into a normal form.

 In this paper we consider bivectors in arbitrary dimensions, and particularly
their properties under Lorentz transformations, with the aim to algebraically
classify the Weyl tensor in higher dimensions 
(based, in part, on the eigenbivector problem).
We first define the Weyl bivector operator in any dimension while
keeping in mind the boost weight (b.w.) decomposition in order to utilize the
algebraic classification of  \cite{class}.  We then refine the classification of  \cite{class}
in terms of the irreducible representations of the spins.  This enables us to define
several algebraically special cases; for example, a number of
special subcases in 5d are presented in section \ref{sect6}.
In particular, the refinement is particularly
amenable to cases where the Weyl tensor has (additional) symmetries. In the final sections
we present some results concerning algebraically special Weyl tensors with symmetries
and make some brief comments on eigenvalue problems and 
$\mathcal{I}$-non-degeneracy. We note that brief reviews of the algebraic classification
of the Weyl tensor of \cite{class} and curvature operators introduced in  \cite{inv}
are given in the last two appendices.

\section{The Bivector operator}
Given a vector basis ${\bf k}^{\mu}$ we can define a set of (simple) bivectors 
\[ {\bf F}^A\equiv {\bf F}^{\mu\nu}={\bf F}^{[\mu\nu]}={\bf k}^{\mu}\wedge{\bf k}^{\nu}, \]
spanning the space of antisymmetric tensors of rank 2. The interpretation of such a basis is clear: ${\bf F}^A$ spans a 2-dimensional plane defined by the vectors ${\bf k}^\mu$ and ${\bf k}^\nu$. This implies (and follows also directly from the definition) that an antisymmetric tensor $G_{\alpha\beta}$ is a simple bivector if and only if $G_{\alpha[\beta}G_{\gamma\delta]}=0$.

Consider a $d=(2+n)$-dimensional Lorentzian space with the following null-frame $\{ {\mbold\ell}, {\bf n},{\bf m}^i\}$ so that the metric is
\[ 
\d s^2=2{\mbold\ell}{\bf n}+\delta_{ij}{\bf m}^i{\bf m}^j.
\] 
Let us consider the following bivector basis (in that order): 
\[ {\mbold\ell}\wedge{\bf m}^i, \quad  {\mbold\ell}\wedge{\bf n},\quad  {\bf m}^i\wedge{\bf m}^j,\quad   {\bf n}\wedge{\bf m}^j,\] or for short: $[0i]$, $[01]$, $[ij]$, $[1i]$. The Lorentz metric also induces a metric, $\eta_{MN}$, in bivector space. If $m=n(n-1)/2$, then
\[ (\eta_{MN})=\frac 12 \begin{bmatrix} 
0 & 0 & 0 & {\sf 1}_n \\
0 & -1& 0 & 0\\
0 & 0 & {\sf 1}_m & 0\\
{\sf 1}_n & 0 & 0 & 0
\end{bmatrix},
\]
where ${\sf 1}_n$, and ${\sf 1}_m$ are the unit matrices of size $n\times n$ and $m\times m$, respectively, and we have assumed the bivector basis is in the order given above. This metric can then be used to raise and lower bivector indices. 

Let $V\equiv \wedge^2T^*_pM$ be the vector space of bivectors at a point $p$.  Then
consider an operator ${\sf C}=(C^M_{~N}):  V\mapsto V$.  We will assume that it is
symmetric in the sense that $C_{MN}=C_{NM}$.

With these assumptions, the operator ${\sf C}$ can be written on the following $(n+1+m+n)$-block form: 
\beq\label{WeylOperator}
{\sf C}=\begin{bmatrix}
M & \hat{K} & \hat{L} & \hat{H} \\
\check{K}^t & -\Phi & -A^t & -\hat{K}^t \\
\check{L}^t & A & \bar{H} & \hat{L}^t \\
\check{H} & -\check{K} & \check{L} & M^t 
\end{bmatrix}
\eeq
Here, the block matrices $H$ (barred, checked and hatted) are all symmetric. Checked (hatted) matrices correspond to negative (positive) b.w. components.

The eigenbivector problem can now be formulated as follows. A bivector $F^A$ is an eigenbivector of ${\sf C}$ if and only if 
\[ C^M_{~N}F^N=\lambda F^M, \quad \lambda\in\mathbb{C}.\] 
Such eigenbivectors can now be determined using standard results from linear algebra. 

\subsection{Lorentz transformations} 
The Lorentz transformations in $(2+n)$-dimensions, $SO(1,1+n)$ act on bivector space via its adjoint representation, $\Gamma$: 
\[ SO(1,1+n)\overset{\Gamma}{\longrightarrow}SO(1+n,m+n).\]
Therefore, in terms of the Lie algebras: 
\[ \mf{s}\mf{o}(1,1+n)\overset{X}{\longrightarrow}\mf{s}\mf{o}(1+n,m+n).\]
 So the Lie algebra representation $X$, will be antisymmetric with respect to $\eta_{MN}$:
\[ \eta_{MR}X^R_{~N}+\eta_{NR}X^R_{~M}=0.\]

Lorentz transformations consist of boosts, spins and null rotations. 
Let us consider each in turn.

\subsubsection{Boosts}
The boosts, $B:~V\mapsto V$, can be represented by the following matrix: 
\beq
B=\begin{bmatrix} 
e^{\lambda} & 0 & 0 & 0 \\
0 & 1& 0 & 0\\
0 & 0 & {\sf 1}_m & 0\\
0 & 0 & 0 & e^{-\lambda}
\end{bmatrix}.
\eeq
This matrix has boost weight 0, and, if we decompose ${\sf C}$ according to the boost weight, ${\sf C}=\sum_b({\sf C})_b$, the boosts have the following property: 
\[ B({\sf C})_bB^{-1}=e^{b\lambda}({\sf C})_b.\] 

\subsubsection{Spins}
The spins can be considered as a matrix 
\beq 
R=\begin{bmatrix} 
G & 0 & 0 & 0 \\
0 & 1& 0 & 0\\
0 & 0 & \bar{G} & 0\\
0 & 0 & 0 & G
\end{bmatrix},
\eeq
where $G\in SO(n)$ and $\bar{G}$ is its adjoint representation; i.e., for an antisymmetric $n\times n$ matrix $A\in\mf{s}\mf{o}(n)$, so that $A=A^B\epsilon_B$ where $\epsilon_B$ is a basis for $\mf{s}\mf{o}(n)$, $GAG^{-1}=\epsilon_C\bar{G}^C_{~B}A^B$. 

The spins also have boost weight 0. 

\subsubsection{Null rotations}
The null rotations can be considered as follows. For a column vector $z\in \mathbb{R}^n$, define the b.w. $-1$ operator (in the Lie algebra representation of the null rotations):
\beq
\check{N}(z)=\begin{bmatrix} 
0 & 0 & 0 & 0 \\
z^t & 0& 0 & 0\\
-Z^t & 0 & 0 & 0\\
0 & z & Z & 0
\end{bmatrix}, \quad z=(z^i), \quad Z=(Z^k_{~B})=(z_i\delta^k_{~j}-z_j\delta^k_{~i})
\eeq
There is also an analogous b.w. $+1$ operator $\hat{N}(z)$. The null rotations can now be found by exponentiation: $N(z)=e^{\check{N}(z)}$. This can be equated to (which is also the b.w. decomposition):
\beq
N(z)=e^{\check{N}(z)}=1+\check{N}(z)+\frac 12\check{N}(z)^2.
\eeq
Note that the inverse is as follows: 
\[
N^{-1}(z)=N(-z)=e^{-\check{N}(z)}=1-\check{N}(z)+\frac 12\check{N}(z)^2.
\]
The null rotations are the only transformations that mix up the boost weight decomposition. 

\section{The Weyl operator}
Let us henceforth consider the Weyl tensor. For the Weyl tensor we can make the following identifications (indices $B,C,..$ should be understood as indices over $[ij]$) \footnote{Note that the presence of $\tfrac 12$ in front of the bivector metric implies that we have to be a bit careful when lowering/raising indices. Here we will \emph{define} components $C^M_{~N}$ to be the components of the Weyl tensor, which differs from the definition in \cite{HallBook} in 4d.}
: 
\beq
\hat{H}^i_{~j}=C_{0i0j}, && \check{H}^i_{~j}=C_{1i1j},\\
\hat{L}^i_{~B}=C_{0ijk}, && \check{L}^i_{~B}=C_{1ijk}, \\
\hat{K}^i=C_{010i}, && \check{K}^i=-C_{011i}, \\
M^i_{~j}=C_{1i0j}, && \Phi=C_{0101}, \\
A^{B}=C_{01ij}, && \bar{H}^B_{~C}=C_{ijkl}.
\eeq
The Weyl tensor is also traceless and obeys the Bianchi identity: 
\[ C^{\mu}_{~\alpha\mu\beta}=0, \quad C_{\alpha(\beta\mu\nu)}=0.\] 
These conditions translate into conditions on our block matrices. Let us consider each boost weight in turn, and let us use this to express these matrices into irreducible representations of the spins.

Note that only non-positive boost weights are considered. To get the positive boost weights, replace all checked quantities with hatted ones.

\subsection{Boost weight 0 components}
Here we have 
\beq
C_{0101}={C_{0i1}}^i, \quad C_{0i1j}=-\tfrac 12{C_{ikj}}^k+\tfrac 12C_{01ij},\quad C_{i(jkl)}=0. 
\eeq 
 Starting with the latter, this means that the matrix $\bar{H}^B_{~C}$ fulfills the reduced Bianchi identies. It is also symmetric which means that it has the same symmetries as an $n$-dimensional Riemann tensor. Hence, we can split this into irreducible parts over $SO(n)$ using the ``Weyl tensor'', ``trace-free Ricci'' and ``Ricci scalar'' as follows ($n>2$): 
\beq
\bar{H}_{BC}&=&\bar{C}_{ijkl}+\frac{2}{n-2}\left(\delta_{i[k}\bar{R}_{l]j}-\delta_{j[k}\bar{R}_{l]i}\right)-\frac{2}{(n-1)(n-2)}\bar{R}\delta_{i[k}\delta_{l]j}, \\
\bar{R}_{ij}&=&\bar{S}_{ij}+\tfrac 1n\bar{R}\delta_{ij}.
\eeq
The remaining Bianchi identities now imply: 
\beq
M_{ij}&=&-\tfrac{1}{2n}\bar{R}\delta_{ij}-\tfrac 12 \bar{S}_{ij}-\tfrac 12A_{ij}\\
\Phi&=&-\tfrac 12\bar{R}.
\eeq
This means that the b.w. 0 components can be specified using the irreducible compositions above: 
\beq
\underbrace{\bar{R},\qquad \bar{S}_{ij}}_{\tfrac{n(n+1)}2},\qquad \underbrace{A_{ij}}_{\tfrac{n(n-1)}2},\qquad \underbrace{\bar{C}_{ijkl}}_{\tfrac{n(n+1)(n+2)(n-3)}{12}}. 
\eeq
Note that in lower dimensions we have the special cases for the $n$-dimensional Riemann tensor: 
\begin{itemize}
\item{} Dim 4 ($n=2$): $\bar{S}_{ij}=\bar{C}_{ijkl}=0$.
\item{} Dim 5 ($n=3$): $\bar{C}_{ijkl}=0$.
\item{} Dim 6 ($n=4$): $\bar{C}_{ijkl}=\bar{C}^+_{ijkl}+\bar{C}^-_{ijkl}$, 
where $\bar{C}^+$ and $\bar{C}^-$ are the self-dual, and  the anti-self-dual parts of the Weyl tensor, respectively. The same can be done with the antisymmetic tensor $A_{ij}=A^+_{ij}+A^-_{ij}$. 
\end{itemize}

A spin $G\in SO(n)$ acts as follows on the various matrices: 
\beq
(M,\Phi,A,\bar{H})\mapsto (GMG^{-1},\Phi,\bar{G}A,\bar{G}\bar{H}\bar{G}^{-1}).
\eeq
If $C_{\mu\nu\alpha\beta}$ is the Weyl tensor (we recall that Greek indices run over the full spacetime manifold), the type D case is therefore completely characterised in terms 
of a $n$-dimensional Ricci tensor, a Weyl tensor, and an antisymmetric tensor $A_{ij}$, as explained earlier. 
Therefore, let us  first use the spins to diagonalise the ``Ricci tensor'' 
$\bar{R}_{ij}$. This matrix can then be described in terms of the Segre-like notation 
corresponding to its eigenvalues. There are two types of special cases 
worthy of consideration. The first is the usual degeneracy in the eigenvalues which occurs when two, or more, eigenvalues are equal. The other special case happens when a eigenvalue is zero. Using a Segre-like notation, we therefore get the types 
\beq
\bar{R}_{ij}:&&\{1111..\}, \quad \{(11)11..\}, \quad\{(11)(11)...\}\quad \text{etc.}\nonumber\\ 
&&
\{0111..\}, \quad \{0(11)1..\}, \quad\{00(11)...\}\quad \text{etc.},
\eeq
where a zero indicates a zero-eigenvalue. 
Regarding the antisymmetric matrix $A_{ij}$, this must be of even rank. 
A standard result gives the canonical block-diagonal form of $A$:
\beq
A=\mathrm{blockdiag}\left(\begin{bmatrix}
0 & -a_1 \\
a_1 & 0 
\end{bmatrix},\begin{bmatrix}
0 & -a_2 \\
a_2 & 0 
\end{bmatrix}, \cdots,\begin{bmatrix}
0 & -a_k \\
a_k & 0 
\end{bmatrix},0,\cdots, 0\right).
\eeq
Therefore, we see that we can characterise an antisymmetric matrix using the rank; i.e., $2k$. Note that the stabiliser is given by
\[ \underbrace{O(2)\times \cdots\times O(2)}_{k \text{ factors}}\times O(n-2k)\subset O(n).\]
The antisymmetric matrix $A$ may also possess degeneracies allowing for more symmetries. For example, if $a_1=a_2$, then the symmetry $O(2)\times O(2)$ acting on these blocks, is being enhanced to $U(2)$ (under the adjoint action). 

Finally, characterisation of the ``Weyl tensor'' $\bar{C}_{ijkl}$ reduces to characterising the Weyl tensor of the corresponding fictitious $n$-dimensional Riemannian manifold.

Algebraically, we can also consider alignments; e.g., if the antisymmetric tensor 
$A_{ij}$ which can be an eigenvector or not, to $\bar{H}^B_{~C}$. If the bivector $A^B$ is indeed an eigenbivector then this would be a special case.

\subsection{Boost weight -1 components}
For the b.w. $-1$ components we have the following identities: 
\beq
C_{011i}=-{C_{1ji}}^j, \quad C_{1(ijk)}=0.
\label{id-1}\eeq
Let us start with the latter identity which involves the matrix $\check{L}=(\check{L}^i_{~B})$. The index $B$ is an antisymmetric index $[jk]$, therefore, we can consider the tensor
\[ \check{L}^i_{~jk}=-\check{L}^i_{~kj}\]
Such tensors have been classified (see, e.g., \cite{TV}). Furthermore, 
the Bianchi identity implies that $\check{L}_{(ikj)}=0$, which gives:
\begin{enumerate} 
\item{} For $n=2$, there exists a vector $\check{v}_i$ such that 
\[ \check{L}^i_{~jk}=\delta^i_{~j}\check{v}_k-\delta^i_{~k}\check{v}_j.\] 
\item{} For $n\geq 3$, there exists an irreducible and orthogonal decomposition so that 
\[ \check{L}^i_{~jk}=\delta^i_{~j}\check{v}_k-\delta^i_{~k}\check{v}_j+\check{T}^i_{~jk},\]
where $\check{T}_{(ijk)}=0$, and $\check{T}^i_{~ji}=0$. 
\end{enumerate}
The first of the identities (\ref{id-1}) now gives
\beq
\check{K}^i=-(n-1)\check{v}^i.
\eeq
Therefore, the boost weight -1 components can be specified using the irreducible decomposition: 
\beq
\underbrace{\check{v}^i}_n, \qquad \underbrace{\check{T}^i_{~jk}}_{{\tfrac{n(n^2-4)}{3}}} \quad \text{where}\quad \check{T}^i_{~(jk)}=\check{T}_{(ijk)}=\check{T}^i_{~ji}=0.
\eeq
Also note that 
\beq
\check{K}^i=0 &\Leftrightarrow& \check{L}^i_{~ji}=0, \\
\check{T}^i_{~jk}=0 & \Leftrightarrow & \check{L}_{ijk}\check{L}^{ijk}=\tfrac{2}{n-1}\check{L}^{ji}_{~~j}\check{L}^k_{~ik}.
\eeq
These conditions can be used to characterise the b.w. $-1$ components into a type (A) case (for which $\check{K}=0$) and a type (B) case (for which $\check{T}=0$). 

\subsection{Boost weight -2 components}
Here we have that
\beq
{C_{1i1}}^i=0. 
\eeq
This identity implies that the matrix $\check{H}$ is traceless: 
\beq
\Tr \check{H}=0.
\eeq
The matrix $\check{H}$ can be characterised in terms of its eigenvalues, $\check{\lambda}_i$. These eigenvalues fulfill $\sum_i\check{\lambda}_i=0$.

\section{The algebraic classification} 
Let us consider the classification in \cite{class}, and investigate the different
algebraic types in turn.  In general, there will also be algebraically special cases of type G as $\hat{H}$ can have degenerate cases (e.g., $\{ (11)1\}$) or there might be alignment of the various blocks.

\subsection{Type I}
The tensor $C_{\mu\nu\alpha\beta}$ is of type I if and only if there exists a null 
frame such that the operator ${\sf C}$ takes the form:
\beq
{\sf C}=\begin{bmatrix}
M & \hat{K} & \hat{L} & 0 \\
\check{K}^t & -\Phi & -A^t & -\hat{K}^t \\
\check{L}^t & A & \bar{H} & \hat{L}^t \\
\check{H} & -\check{K} & \check{L} & M^t 
\end{bmatrix}
\eeq
If  $C_{\mu\nu\alpha\beta}$ is the Weyl tensor, 
there would be two subcases for which:\footnote{Note that $\hat{K}=0$ and $\hat{T}$ are invariant and unambiguous statements when only one WAND is present, but not necessarily when there are several WANDs (i.e., for type I$_i$ \cite{class}).} 
\begin{itemize}
\item{} Type I(A): $\hat{K}=0$.
\item{} Type I(B): $\hat{T}=0$.
\end{itemize}
 More precisely, we shall refer to these as subcases  $\hat{A}$ 
and $\hat{B}$ to differentiate
the subcases in subsection \ref{typeIII}; however, we shall omit the quantifiers if 
the context is clear.

If the type I Weyl tensor is in the aligned null-frame, then criteria for the Weyl tensor being in each of these subcases can be given as follows:
\begin{itemize}
\item{} Type I(A) $\Leftrightarrow$ $C^{i}_{ji0}=0$.
\item{} Type I(B) $\Leftrightarrow$ $C_{ijk0}C^{ijk}_{\phantom{ijk}0}=\tfrac{2}{n-1}C^{ji}_{\phantom{ji}j0}C^k_{\phantom{k}ik0}$. 
\end{itemize}

\subsection{Type II}
The tensor $C_{\mu\nu\alpha\beta}$ is of type II if and only if there 
exists a null frame such that the operator ${\sf C}$ takes the form:
\beq
{\sf C}=\begin{bmatrix}
M & 0 & 0 & 0 \\
\check{K}^t & -\Phi & -A^t & 0 \\
\check{L}^t & A & \bar{H} & 0 \\
\check{H} & -\check{K} & \check{L} & M^t 
\end{bmatrix}
\eeq
Then there will be algebraic subcases according to whether some of the irreducible 
components of b.w. 0 are zero or not. A complete characterisation of all such subcases 
is very involved in its full generality. However, a rough classification in terms of the vanishing the irreducible components under spins can be made: 
\begin{itemize}
\item{} Type II(a): $\bar{R}=0$
\item{} Type II(b): $\bar{S}_{ij}=0$
\item{} Type II(c): $\bar{C}_{ijkl}=0$.
\item{} Type II(d): $A=0$
\end{itemize}
Note that we can also have a combination of these; for example, type II(ac), which means that $\bar{R}=0$ and $\bar{C}=0$.

\subsection{Type D}
The tensor $C_{\mu\nu\alpha\beta}$ is of type D if and only if there exists a null 
frame such that the operator ${\sf C}$ takes the form:
\beq
{\sf C}=\begin{bmatrix}
M & 0 & 0 & 0 \\
0 & -\Phi & -A^t & 0 \\
0 & A & \bar{H} & 0 \\
0 & 0 & 0 & M^t 
\end{bmatrix}
\eeq
Here all Lorentz transformations have been utilised, \emph{except} for the spins. 
In addition, we note that type D tensors are invariant under boosts. 

As in the case of type II, we can have subcases as follows:
\begin{itemize}
\item{} Type D(a): $\bar{R}=0$
\item{} Type D(b): $\bar{S}_{ij}=0$
\item{} Type D(c): $\bar{C}_{ijkl}=0$
\item{} Type D(d): $A=0$.
\end{itemize}

 Note that in principle analogous algebraically special subcases exist for type I etc.
For example, an algebraically special type I(A) can also obey further conditions like 
 (a), (b), (c) or (d) (e.g., type I(Aad), etc).

\subsection{Type III}
\label{typeIII}
The tensor $C_{\mu\nu\alpha\beta}$ is of type III if and only if 
there exists a null frame such that the operator ${\sf C}$ takes the form:
\beq
{\sf C}=\begin{bmatrix}
0 & 0 & 0 & 0 \\
\check{K}^t & 0 & 0 & 0 \\
\check{L}^t & 0 & 0 & 0 \\
\check{H} & -\check{K} & \check{L} & 0 
\end{bmatrix}
\eeq
If  $C_{\mu\nu\alpha\beta}$ is the Weyl tensor, then the orthogonal 
decomposition of $\check{L}$ defines two subcases:
\begin{itemize}
\item{} Type III(A): $\check{K}=0$.
\item{} Type III(B): $\check{T}=0$.
\end{itemize}
More precisely, types  III(\v{A}) and  III(\v{B}).
Again, these cases can be given in terms of conditions on the Weyl tensor:
\begin{itemize}
\item{} Type III(A) $\Leftrightarrow$ $C^{i}_{ji1}=0$.
\item{} Type III(B) $\Leftrightarrow$ $C_{ijk1}C^{ijk}_{\phantom{ijk}1}=\tfrac{2}{n-1}C^{ji}_{\phantom{ji}j1}C^k_{\phantom{k}ik1}$. 
\end{itemize}

\subsection{Type N}
The tensor $C_{\mu\nu\alpha\beta}$ is of type N if and only if there exists a null frame such that the operator ${\sf C}$ takes the form
\beq
{\sf C}=\begin{bmatrix}
0 & 0 & 0 & 0 \\
0 & 0 & 0 & 0 \\
0 & 0 & 0 & 0 \\
\check{H} & 0 & 0 & 0 
\end{bmatrix}
\eeq

A type N tensor can be completely classified as follows.  The matrix $\check{H}$ is
symmetric, so by using the spins we can therefore diagonalise this matrix completely.  The
type N tensor is therefore characterised by the eigenvalues of the matrix $\check{H}$.  If
$C_{\mu\nu\alpha\beta}$ is the Weyl tensor, this matrix is traceless and, using a
Segre-like notation, we get the following possibilities in low dimensions:
\begin{itemize}
\item{} Dim 4: $\{11\}$ 
\item{} Dim 5: $\{111\}$, $\{(11)1\}$,  $\{110\}$
\item{} Dim 6: $\{1111\}$, $\{(11)11\}$, $\{(111)1\}$, $\{(11)(11)\}$, $\{1110\}$, $\{(11)10\}$,  $\{1100\}$  
\end{itemize}

\section{Dimension 4 ($n=2$)}
Let us consider the special case of dimension 4 and show how it reduces to the 
standard analysis. In 4 dimensions, the Weyl operator can always be put in type I form by using a null rotation (hence, $\hat{H}=0$). Furthermore, the irreducible representations under the spins are: 
\[ \hat{v}^i, \quad \bar{R},\quad A, \quad \check{v}^i, \quad \check{H}.\] 
We notice that seemingly there are a total of 8 components; however, we still have the 
unused freedom of one spin, one boost and two null-rotations. 
In each of the algebraically special cases we can use these to simplify the Weyl 
tensor even further. Since the 4 dimensional case is well-known, 
let us only consider type D and type III for illustration. 
 \subsection{Type D}
For $n=2$, the Weyl tensor reduces to specifying two scalars, namely $\bar{R}$ and $A_{34}$. We now get
\beq
M=\begin{bmatrix}
-\tfrac 14\bar{R} & -\tfrac 12A_{34} \\
\tfrac 12A_{34} & -\tfrac 14\bar{R} 
\end{bmatrix}, \quad 
\begin{bmatrix}
-\Phi & -A^t  \\
A & \bar{H}
\end{bmatrix}=\begin{bmatrix}
\tfrac 12\bar{R} & -A_{34} \\
A_{34} & \tfrac 12\bar{R} 
\end{bmatrix};
\eeq
consequently, the Weyl operator ${\sf C}$ has eigenvalues:
\beq
\lambda_{1,2}=\lambda_{3,4}=-\frac 14(\bar{R}\pm 2iA_{34}), \quad \lambda_{5,6}=\frac 12(\bar{R}\pm 2iA_{34}).
\eeq
We note that this is in agreement with the standard type D analysis in 4 dimenions (see \cite{HallBook}). The type D case is boost invariant, and also invariant under spins, consequently the isotropy is 2-dimensional. 

The two subcases $A_{34}=0$ and $\bar{R}=0$ (type D(d) and D(a), respectively) 
are in 4 dimensions referred to the purely ``electric'' and ``magnetic'' cases, 
respectively. In 4 dimensions, there is a duality relation, $\star$, which interchanges 
these two cases; i.e., $C\mapsto \star C$ interchanges the electic and magnetic parts.

\subsection{Type III}
In dimension 4, only case III(B) exists, which is idential to the type III general case. 
Therefore, the type III case has only non-trivial $\check{v}^i$ and $\check{H}_{ij}$. However, we still have unused freedom in spins, null-rotations, and boosts. 

Using the equations for null-rotations in Appendix \ref{AppNull} when applying a null-rotation, 
the $\check{H}_{ij}$ transform as (for type III case):
\[ \underline{\check{H}}_{ij}= \check{H}_{ij}+4\check{v}_{(i}z_{j)}-2\delta_{ij}\check{v}^kz_{k}.\] 
Let us see if we can set $\underline{\check{H}}_{ij}=0$ using a null-rotation. This reduces to requiring: 
\beq
0&= &\check{H}_{33}+2(\check{v}_3z_3-\check{v}_4z_4) \nonumber \\
0&= &\check{H}_{34}+2(\check{v}_4z_3+\check{v}_3z_4).
\eeq
We can always find a $z^i$ solving these equations provided that $\check{v}_3^2+\check{v}_4^2\neq 0$. Consequently, for a proper type III spacetime, we can always use a null rotation so that $\check{H}_{ij}=0$. The remaining spin and boost can then be used to set $\check{v}^3=1$, $\check{v}^4=0$. This is a well-known result in 4 dimensions.

\section{Dimension 5 ($n=3$)}
\label{sect6}
The 5-dimensional case is considerably more difficult than the 4-dimensional case. 
The complexity drastically increases and hence the number of special cases also 
increases. However, the 5 dimensional case is still managable and some 
simplifications occur (compared to the general case). Most notably, $\bar{C}_{ijkl}=0$, 
and $\check{T}^i_{jk}$ can be written, using a  matrix $\check{n}_{ij}$, as follows 
(similarly for $\hat{T}^i_{~jk}$): 
\beq 
\check{T}^i_{~jk}=\varepsilon_{jkl}\check{n}^{li},
\label{T=n}\eeq
where the conditions on $\check{T}^i_{~jk}$ imply that $\check{n}^{ij}$ is symmetric and 
trace-free. 

Therefore, we have the following components in dimension 5:
\begin{itemize}
\item{} b.w. $+2$: $\hat{H}_{ij}$
\item{} b.w. $+1$: $\hat{v}^i$, $\hat{n}_{ij}$
\item{} b.w. $0$: $\bar{R}$, $\bar{S}_{ij}$, $A_{ij}$
\item{} b.w. $-1$: $\check{v}^i$, $\check{n}_{ij}$
\item{} b.w. $-2$: $\check{H}_{ij}$
\end{itemize}

Let us consider here the following order of the spatial bivectors: 
\[ [45],~[53],~[34],\]
and consider some of the special cases in dimension 5.
\subsection{Type D}
For a type D Weyl tensor only the following components can be non-zero:
\[ \bar{R}, \quad \bar{S}^i_{~ j},\quad  A_{ij},\]
where $i,j=3,4,5$. 
Let us use the spins to diagonalise $(\bar{S}^i_{~j})=\mathrm{diag}(S_{33},S_{44},S_{55})$. Without any further assumptions, the Weyl blocks take the form:
\beq
M&=&\begin{bmatrix}
-\tfrac 16\bar{R}-\tfrac 12S_{33} & -\tfrac 12A_{34} & \tfrac 12A_{53} \\
\tfrac 12A_{34} & -\tfrac 16\bar{R}- \tfrac 12S_{44} & -\tfrac 12A_{45} \\
-\tfrac 12A_{53} & \tfrac 12A_{45} & -\tfrac 16\bar{R}-\tfrac 12S_{55} 
\end{bmatrix}, \nonumber \\
\begin{bmatrix}
-\Phi & -A^t  \\
A & \bar{H}
\end{bmatrix}
&=& \begin{bmatrix}
\tfrac 12\bar{R} & -A_{45} & -A_{53} & -A_{34} \\
A_{45} & \tfrac 16\bar{R}-S_{33} & 0 & 0 \\
A_{53} & 0 &   \tfrac 16\bar{R}-S_{44} & 0 \\
A_{34} & 0 & 0 &  \tfrac 16\bar{R}-S_{55} 
\end{bmatrix}
\eeq
The general type D tensor thus has this canonical form. 

There are two special cases where we can use the extra symmetry to get the simplified  
canonical form:
\begin{itemize}
\item{} (i) $S_{33}=S_{44}=-\tfrac 12S_{55}$: $A_{53}=0$
\item{} (ii) $S_{33}=S_{44}=S_{55}=0$: $A_{53}=A_{45}=0$. 
\end{itemize}
We note that case (ii) will, without further assumptions, be invariant under spatial
rotations in the $[34]$-plane (in addition to the boost).  Assuming, in addition, that
$A_{45}=0$, then case (i)  is also invariant under a rotation in the $[34]$-plane.
Assuming that $A_{ij}$ vanishes completely (hence, type D(d)), we note that case (ii) enjoys the full
invariance under the spins (i.e., $SO(3)$).

\subsection{Type III}
A type III Weyl tensor can have the following non-zero components:
\[ \check{v}^i, \quad \check{n}_{ij}, \quad \check{H}_{ij}\]
Generally, we can use the spins to diagonalise $\check{n}^{ij}$. 
Therefore, the general case is  $\{111\}$ (all eigenvalues different), with the 
special cases $\{(11)1\}$, $\{110\}$ and $\{000\}$ (the latter case is, of course, 
type III(B) ). Furthermore, in the general case, the vector $\check{v}^i$ needs not 
be aligned with the eigenvectors of $\check{n}^{ij}$ (nor $\check{H}_{ij}$). 
There would consequently be special cases where $\check{v}^i$ is an eigenvector of 
$\check{n}^{ij}$ (and $\check{H}_{ij}$). Therefore, unlike in dimension 4 for which 
there is only one case, the type III case in 5 dimensions have a wealth of subcases. It is possible to explicitly delineate all algebraically special cases; 
this will be done elsewhere \cite{dim5}.

\subsection{ Weyl tensors with symmetry}
Let us consider the case of dimension 5 ($n=3$), where we impose certain (additional) 
symmetries on the Weyl tensor. The isotropy must be a subgroup of $SO(1,4)$, of which there are numerous subgroups. If the isotropy consists of a boost, then it must be of type D. 

Let us concentrate on groups $H$ such that $H\subset SO(4)\subset SO(1,4)$. 
These groups would therefore have spacelike orbits. The group $SO(4)$ is not simple and can be considered as:
\[ SO(4)\cong\frac{SU(2)\times SU(2)}{\mathbb{Z}_2}. \] 
Using the quarternions, $\mathbb{H}$, we can consider the action of $SO(4)$ 
as the action of the unit quaternions, $\mathbb{H}_1\cong SU(2)$, on a vector 
${\sf v}\in \mathbb{H}$ ($\sim \mathbb{R}^4$ as a vector space) as follows:
\beq
{\sf v}\mapsto q_1{\sf v}q_2^{-1}, \qquad (q_1,q_2)\in\mathbb{H}_1\times\mathbb{H}_1.
\eeq
Note that the  ``diagonal action'', ${\sf v}\mapsto q_1{\sf v}q_1^{-1}$, leaves the vector ${\sf v}=1$ invariant; hence, this is the standard $SO(3)\subset SO(4)$. 

In the following we will consider various subgroups of $SO(4)$ that can occur and 
see what restrictions these impose on the form of the Weyl tensor in 5 dimensions. 
However, note that these are not all of the possible subgroups. 

\subsubsection{$SO(2)$}
Here, defining $\kappa=\diag(-2,1,1)$, the Weyl tensor can be chosen to have the following form: 
\beq &&\check{H}=\check{\lambda}_1\kappa, \quad \check{n}=\check{\lambda}_2\kappa, \quad \check{v}^i=\delta^i_3 \check{v}, \nonumber \\
&& A_B=(A_{45},0,0), \quad \bar{S}=s\kappa, \quad \bar{R}, \\
 &&\hat{H}=\hat{\lambda}_1\kappa, \quad \hat{n}=\hat{\lambda}_2\kappa, \quad \hat{v}^i=\delta^i_3 \hat{v}, \nonumber
\eeq
\subsubsection{$SO(2)\times SO(2)$} For this case, the Weyl tensor can be chosen to be of the form ($\kappa=\diag(-2,1,1)$):
\beq &&\check{H}=\hat{H}={\lambda}_1\kappa, \quad \check{n}=\hat{n}={\lambda}_2\kappa, \quad \check{v}^i=\hat{v}^i=0, \nonumber \\
&& A_B=0, \quad \bar{S}=s\kappa, \quad \bar{R}=\tfrac 32(s+2\lambda_1).
\eeq
\subsubsection{$SO(3)$}
Here, there is only one independent component, namely $\bar{R}$. Hence, this is automatically of type D. 
\subsubsection{$SU(2)$}
Defining $D=\diag(\lambda_1,\lambda_2,\lambda_3)$ where $\lambda_1+\lambda_2+\lambda_3=0$, the Weyl tensor can be chosen to be of the following form:
\beq &&\check{H}=\hat{H}=D, \quad \check{n}=\hat{n}=\sqrt{2}D, \quad \check{v}^i=\hat{v}^i=0, \nonumber \\
&& A_B=0, \quad \bar{S}=-2D, \quad \bar{R}=0.
\eeq
\subsubsection{$U(2)$} Again defining $\kappa=\diag(-2,1,1)$, the Weyl tensor can be chosen to have the following form: 
\beq &&\check{H}=\hat{H}=\lambda\kappa, \quad \check{n}=\hat{n}=\sqrt{2}\lambda\kappa, \quad \check{v}^i=\hat{v}^i=0, \nonumber \\
&& A_B=0, \quad \bar{S}=-2\lambda\kappa, \quad \bar{R}=0.
\eeq
\subsubsection{$SO(4)$}
This is the conformally flat case; hence, this is of type $O$. 

\subsection{Examples}
\paragraph{An Einstein space:}
Let us consider a 5-dimensional example and determine the Weyl type of this metric. 
We will consider the particular Kundt metric \cite{CFH}:
\beq \d s^2&=&2\d u\left[\d v+\sigma v^2\d u+\alpha v(\d x+\sin y\d z)\right] \nonumber \\
&& +a^2(\d x+\sin y\d z)^2+b^2(\d y^2+\sin^2y\d z^2).
\eeq
By chosing 
\[ \sigma=\frac{a^2}{4b^4}, \quad \alpha=\frac{\sqrt{2(a^2-b^2)}}{b^2}, \]
this is an Einstein space with $R_{\mu\nu}=[(2b^2-a^2)/(2b^4)]g_{\mu\nu}$. 
By using the standard Kundt frame, 
\beq
&& {\ell}= \d u, \quad {\bf n}=\d v+\sigma v^2\d u+\alpha v(\d x+\sin y\d z), \nonumber \\
&& {\bf m}^3=a(\d x+\sin y\d z), \quad {\bf m}^4=b\d y, \quad {\bf m}^5=b\sin y\d z, 
\eeq
the Weyl tensor is seen to have b.w. $0$, $-1$ and $-2$ terms. However, we still have 
some null-rotations we can use to try to simplify the Weyl tensor further. 
Indeed, performing the null-rotation as follows:
\beq 
&&\tilde{\ell}=\ell, \quad \tilde{\bf n}={\bf n}-\tfrac{\zeta^2}{2}{\bf \ell}-\zeta{\bf m}^3, \nonumber \\
&& \tilde{\bf m}^3={\bf m}^3+\zeta{\ell}, \quad \tilde{\bf m}^4={\bf m}^4, \quad \tilde{\bf m}^5={\bf m}^5,
\eeq
where $\zeta=v\alpha$, brings the Weyl tensor into a type D form. In the tilded null-frame, the only non-zero components of the Weyl tensor are:
\beq
&& \bar{R}= \frac{a^2+2b^2}{4b^4}, \quad A_{45}=-\frac{a}{2b^4}\sqrt{2(a^2-b^2)}, \nonumber \\
&& \quad\bar{S}_{ij}=-\frac{(a^2-b^2)}{3b^4}\diag\left(-2,1,1\right).
\eeq
Consequently, we get the result: 
\begin{enumerate}
\item{} $a^2\neq b^2$: Type D with $SO(2)$ isotropy.
\item{} $a^2=b^2$: Type D with $SO(3)$ isotropy. 
\end{enumerate}
Indeed, one can easily check that these isotropies of the Weyl tensor, including the boost-isotropy, 
correspond to actual Killing vectors of the spacetime (in fact, the metric above is also space-time homogeneous)
\paragraph{A 5d spatially homogeneous cosmology:}
Let us consider the following metric:
\beq
\d s^2= -\d t^2+a(t)^2\Big{[}e^{-2w}\left(\d x+\tfrac 12(y\d z-z\d y)\right)^2\qquad &&\nonumber \\
 +e^{-w}\left(\d y^2+\d z^2\right)+\d w^2\Big{]}.&&
\label{HC2}\eeq
Using the obvious orthonormal frame, we get the following results:
\beq &&\check{H}=\hat{H}=\lambda\kappa, \quad \check{n}=\hat{n}=\sqrt{2}\lambda\kappa, \quad \check{v}^i=\hat{v}^i=0, \nonumber \\
&& A_B=0, \quad \bar{S}=-2\lambda\kappa, \quad \bar{R}=0.
\eeq
Consequently, the Weyl tensor possesses an $U(2)$ isotropy.  Indeed, the spatial sections,
for fixed $t$, is the complex hyperbolic space $\mathbb{H}_{\mathbb{C}}^2$, which can be
considered as the homogeneous space $\frac{SU(1,2)}{U(2)}$.  Consequently, the spacetime
will have an isotropy group $U(2)$ which, of course, the Weyl tensor will inherit.  This
spacetime is thus an example of a spacetime where the Weyl tensor has a
$U(2)$-isotropy.

Let us also check the eigenvalues of the Weyl tensor and see what kind of projection 
operator we can obtain in this case. Due to the simple structure of the Weyl tensor we can easily find the eigenvectors/values of the Weyl operator in this case. Using the orthonormal frame, the eigenvalues and eigenvectors are: 
\begin{itemize}
\item{} $-4\lambda$: $[45]+[x3]$
\item{} $2\lambda$: $[53]+[x4]$, $[34]+[x5]$
\item{} $0$: All bivectors orthogonal to the bivectors above.
\end{itemize} 
If we therefore consider the projector, $\bot$, associated with the eigenvalue 
$-4\lambda$, we note that this can be written: 
\[ \bot^{\alpha\beta}_{\phantom{\alpha\beta}\mu\nu}=F^{\alpha\beta}F_{\mu\nu},\quad  \text{where}\quad  (1/2)F_{\mu\nu}{\mbold\omega}^\mu\wedge{\mbold\omega}^{\nu}={\mbold\omega}^4\wedge{\mbold\omega}^{5}+{\mbold\omega}^x\wedge{\mbold\omega}^3.\] 
We observe that by a contraction we obtain: 
\[ \bot^{\alpha\beta}_{\phantom{\alpha\beta}\mu\beta}=\diag(0,1,1,1,1),\] which is a curvature projector of type $\{1,(1111)\}$. Consequently, in the language of \cite{inv}, \emph{the metric (\ref{HC2}) is an $\mathcal{I}$-non-degenerate metric}. In fact, \emph{all} non-conformally flat 5d metrics with $U(2)$ isotropy of the Weyl tensor are $\mathcal{I}$-non-degenerate. 

\section{Weyl tensors with large symmetry groups}
\label{sect:symmetry}
In this section we will state some results with regards to the Weyl tensor when 
the Weyl tensor possesses a relatively large isotropy group with spacelike orbits.
First, let us state a well-known result regarding almost maximal symmetry:
\begin{thm}
For a $d$-dimensional spacetime, if the Weyl tensor has an $SO(d-1)$ symmetry, then the Weyl tensor vanishes and consequently it is of Weyl type $O$. 
\end{thm} 
\begin{proof}
This proof is well-known and can easily be obtained using the Weyl operator in the orthonormal frame, see Appendix \ref{AppOrtho}. 
\end{proof}
This means that, for $d>3$, all such spacetimes are conformally flat. Examples of such 
spacetimes would be the higher-dimensional Friedmann-Robertson-Walker models. 

\begin{thm}
If the isotropy is $SO(d-2)$ (and $d>4$), then the spacetime necessarily is of Weyl type D(bcd) and has only one independent component, namely $\bar{R}$. 
\end{thm}
\begin{proof}
We can assume this group acts on the basis vectors ${\bf m}^i$. Consequently, 
all matrices need to be completely symmetric under the full set of spins and this 
implies that $\hat{H}=\check{H}=\bar{S}=0$. It is also clear that $\check{v}=\hat{v}=0$ 
and that $\check{T}=\hat{T}=0$. Moreover, $\bar{C}_{ijkl}$ is 
a ``Weyl tensor'' with maximal symmetry; hence, this is zero also. 
For $A_{ij}$ and $d>4$, we also need $A_{ij}=0$. The only scalar is $\bar{R}$,
which is clearly invariant and therefore the only component that can be non-zero.
\end{proof}
Note that for $d=4$, the symmetry $SO(2)$ also requires the Weyl tensor to be of type D; 
however, $A_{45}$ needs not be zero. Therefore, in 4 dimensions the Weyl tensor may have 2 non-zero components.
\begin{cor}
The eigenvalues of a Weyl tensor with $SO(d-2)$ symmetry are:
\begin{itemize}
\item{} $d=4$: 
\[ -\frac 14(\bar{R}\pm 2iA_{34})~~[\times 2], \quad \frac 12(\bar{R}\pm 2iA_{34}).\]
\item{} $d>4$:
\[ -\frac{1}{2(d-2)}\bar{R}~~[\times 2(d-2) ], \quad \frac 12\bar{R}, \quad \frac{1}{(d-2)(d-3)}\bar{R}~~[\times\tfrac{(d-2)(d-3)}2 ]\] 
\end{itemize}
\end{cor}
We notice that in this case, as long as $\bar{R}\neq 0$, there will be one projection operator of the block form (here, ${\sf 0}$ is the square zero-matrix) 
\[ \bot=\diag( {\sf 0}_n,1,{\sf 0}_m,{\sf 0}_n),\] 
as long as $d>4$. Consequently, by a contraction $P^\mu_{~\nu}\equiv (\bot)^{\mu\alpha}_{\phantom{\mu\alpha}\nu\alpha}$, we get projectors of type $\{(1,1)(11...1)\}$.

For a smaller isotropy group we have the following result:
\begin{thm}
If the isotropy is $SO(d-3)$ (and $d>5$), then the Weyl tensor can be put in the following form ($\kappa=\diag(d-3,-1,-1,\dots,-1)$):
\beq
&&\check{H}=\check{\lambda}_1\kappa, \quad \check{T}=0, \quad \check{v}^i=\delta^i_3 \check{v}, \nonumber \\
&& A_B=0, \quad \bar{S}=s\kappa, \quad \bar{R},\quad \bar{C}_{ijkl}=0 \\
 &&\hat{H}=\hat{\lambda}_1\kappa, \quad \hat{T}=0, \quad \hat{v}^i=\delta^i_3 \hat{v} \nonumber
\eeq
\end{thm}
\begin{proof}
We choose the action so that it leaves ${\bf m}^3$ invariant while acting 
on ${\bf m}^i$, $i=4,...,d$ in the standard way as a vector. We now see immediately
that $\check{H}$, $\hat{H}$, $\bar{S}$, $\check{v}$ and $\hat{v}$ have the claimed form. 
For spatial bivectors the action of $SO(d-3)$ will be as follows: $[3i]$ as a vector,
$[ij]$ as bivector, for $i,j=4,...,d$. Consequently, since $d>5$, then $A_B=0$. 
For the tensor $\hat{T}$, we get $\hat{T}^3_{~3i}=0$, $\hat{T}^3_{~ij}=0$, 
$\hat{T}^k_{~ij}=0$, and $\hat{T}^k_{~3i}\propto\delta^k_{~i}$. However, 
since $\hat{T}^i_{~3i}=0$, we thus get $\hat{T}=0$. Similarly, $\check{T}=0$. For the ``Weyl tensor'' $\bar{C}_{ijkl}$, this is the $(d-2)$-dimensional Weyl tensor with $SO(d-3)$ symmetry; consequently, $\bar{C}_{ijkl}=0$. 
\end{proof}
Finally, we should point out that we still have null rotations and a boost left which can be used to simplify the Weyl tensor even further. 
\begin{thm}
Assume that the spacetime has odd dimension, $d=2k+1>4$. Assume also that the Weyl tensor is non-zero with  isotropy $U(k)$. Then the spacetime is locally conformal to 
\[ -\d t^2+\d \sigma^2, \] 
where $\d \sigma^2$ is one of the following Riemannian spaces:
\begin{enumerate}
\item{} complex projective space, $\mathbb{C}\mathbb{P}^k$, with the Fubini-Study metric.
\item{} complex hyperbolic space, $\mathbb{H}_{\mathbb{C}}^k$, with the Bargmann metric.
\end{enumerate}\label{thmU(k)}
\end{thm}
\begin{proof}
First we note that the group $U(k)$ acts irreducibly on the spatial tangent 
space spanned by $\{ {\mbold \omega}^x, {\bf m}^i\}$ (orthonormal frame). 
Using the bivector operator for the orthonormal frame (see Appendix \ref{AppOF}), we find that the Weyl tensor components involving $t$ have to be of the following form:
\[ C_{txtx}=\lambda, \quad C_{titj}=\lambda\delta_{ij}, \quad \text{other} ~C_{t\alpha\mu\nu}=0.\]
However, $C^{t\mu}_{\phantom{t\mu}t\mu}=0$, which means that $\lambda=0$ and thus $C_{t\alpha\mu\nu}=0$. Consequently, the Weyl tensor is purely spatial with symmetry $U(k)$. A standard result is that this Weyl tensor only has one independent component. The Weyl tensor can therefore be written $C=\phi^2\tilde{C}$, where $\phi$ is some function and $\tilde{C}$ is a constant tensor in the orthonormal frame. This tensor is $U(k)$-symmetric and has only one component; therefore, the Weyl tensor has to be proportional to one of the 2 spaces given (the proportionality factor, $\phi^2$, will be related to the conformal factor).  
\end{proof}
In the above theorem the group $U(k)$ acts irreducibly on the spatial vectors. Consequently, if we consider the rank 2 tensor $T^{\mu}_{~\nu}\equiv C^{\mu\alpha\beta\gamma}C_{\nu\alpha\beta\gamma}$, it must necessarily have Segre type $\{1,(11...1)\}$ or $\{(1,11...1)\}$. By the proof we see that $T^t_{~t}=0$ while $T^{\mu}_{~\mu}\neq 0 $  since it is a sum of squares. Therefore, the tensor $T^{\mu}_{~\nu}$ has Segre type $\{1,(11...1)\}$, and hence:
\begin{cor}
A spacetime for which the Weyl tensor fulfills the conditions in Theorem \ref{thmU(k)} is $\mathcal{I}$-non-degenerate.
\end{cor}
\section{Summary}

In this paper we have considered bivectors and defined 
the Weyl bivector operator in arbitrary dimensions. 
We have then utilized the Weyl bivector
operator and the boost weight decomposition of the Weyl tensor
and consequently refined the algebraic classification of  \cite{class}
in terms of the irreducible representations of the spins.  Various algebraically
special cases can now  easily be defined.
In Tables \ref{dim4}, \ref{dim5}, and \ref{dim>6} we have summarised the classification of
the Weyl tensor and its independent components in dimensions $d=4$, $d=5$ and $d\geq 6$,
respectively.

All of the algebraically special cases in 5d can be delineated precisely; this will be done in a future work \cite{dim5}.
\newpage
\begin{table}[ht]
\begin{tabular}{|r|l|l|}
\hline 
b.w. & Ind. Components & Weyl components \\
\hline 
$+ 2$ & $\hat{H}_{ij}$  & $ C_{0i0j}=\hat{H}_{ij}$ \\
$+1$  & $\hat{v}_i$ & 
$C_{0ijk}=\delta_{ij}\hat{v}_k-\delta_{ik}\hat{v}_j, \quad  
C_{010i}=-\hat{v}_i $ \\
$0$  & $\bar{R}$, $A_{34}$ & $\begin{cases}C_{1i0j}=-\tfrac 14\bar{R}\delta_{ij}-\tfrac 12A_{ij},\quad C_{01ij}=A_{ij},\\
C_{0101}=-\tfrac 12 \bar{R}, \quad C_{ijkl}=\tfrac 12\bar{R}(\delta_{ik}\delta_{jl}-\delta_{il}\delta_{jk})
\end{cases}$ \\
$-1$ &  $\check{v}_i$ & 
$C_{1ijk}=\delta_{ij}\check{v}_k-\delta_{ik}\check{v}_j, \quad  
C_{011i}=\check{v}_i$ \\
$-2$ &   $\check{H}_{ij}$  & $ C_{0i0j}=\check{H}_{ij}$ \\
\hline
\end{tabular}
\caption{Dimension $d=4$}\label{dim4}
\end{table}
\begin{table}[ht]
\begin{tabular}{|r|l|l|}
\hline 
b.w. & Ind. Components & Weyl components \\
\hline 
$+ 2$ & $\hat{H}_{ij}$  & $ C_{0i0j}=\hat{H}_{ij}$ \\
$+1$  & $\hat{v}_i$, $\hat{n}_{ij}$ & 
$C_{0ijk}=\delta_{ij}\hat{v}_k-\delta_{ik}\hat{v}_j+\varepsilon_{jkl}\hat{n}^{l}_{~i}, \quad  
C_{010i}=-2\hat{v}_i $ \\
$0$  & $\bar{R}$, $\bar{S}_{ij}$, $A_{ij}$ & $\begin{cases}C_{1i0j}=-\tfrac 12\bar{R}_{ij}-\tfrac 12A_{ij},\quad C_{01ij}=A_{ij},\\
C_{0101}=-\tfrac 12 \bar{R}, \quad C_{ijkl}=\bar{R}_{ijkl}
\end{cases}$ \\
$-1$ &  $\check{v}_i$, $\check{n}_{ij}$ & 
$C_{1ijk}=\delta_{ij}\check{v}_k-\delta_{ik}\check{v}_j+\varepsilon_{jkl}\check{n}^{l}_{~i}, \quad  
C_{011i}=2\check{v}_i$ \\
$-2$ &   $\check{H}_{ij}$  & $ C_{0i0j}=\check{H}_{ij}$ \\
\hline
\end{tabular}
\caption{Dimension $d=5$: Here $\bar{R}^k_{~ikj}=\bar{R}_{ij}=\frac 13\bar{R}\delta_{ij}+\bar{S}_{ij}$.} \label{dim5}
\end{table}
\begin{table}[ht]
\begin{tabular}{|r|l|l|}
\hline 
b.w. & Ind. Components & Weyl components \\
\hline 
$+ 2$ & $\hat{H}_{ij}$  & $ C_{0i0j}=\hat{H}_{ij}$ \\
$+1$  & $\hat{v}_i$, $\hat{T}^i_{~jk}$ & 
$C_{0ijk}=\delta_{ij}\hat{v}_k-\delta_{ik}\hat{v}_j+\hat{T}_{ijk}, \quad  
C_{010i}=-(d-3)\hat{v}_i $ \\
$0$  & $\bar{R}$, $\bar{S}_{ij}$, $\bar{C}_{ijkl}$, $A_{ij}$ & $\begin{cases}C_{1i0j}=-\tfrac 12\bar{R}_{ij}-\tfrac 12A_{ij},\quad C_{01ij}=A_{ij},\\
C_{0101}=-\tfrac 12 \bar{R}, \quad C_{ijkl}=\bar{R}_{ijkl}
\end{cases}$ \\
$-1$ &  $\check{v}_i$, $\check{T}^i_{~jk}$ & 
$C_{1ijk}=\delta_{ij}\check{v}_k-\delta_{ik}\check{v}_j+\check{T}_{ijk}, \quad  
C_{011i}=(d-3)\check{v}_i$ \\
$-2$ &   $\check{H}_{ij}$  & $ C_{0i0j}=\check{H}_{ij}$ \\
\hline
\end{tabular}
\caption{Dimension $d\geq 6$: Here  $\bar{R}^k_{~ikj}=\bar{R}_{ij}=\frac 1{d-2}\bar{R}\delta_{ij}+\bar{S}_{ij}$.} \label{dim>6}
\end{table}

\newpage
\section*{Acknowledgements} 
We would like to thank Lode Wylleman and Marcello Ortaggio for comments and suggestions. 

This work was supported, in part, by NSERC of Canada. 

\appendix

\section{Algebraic classification}
Given a covariant tensor $T$ with respect to a generalised Newman-Penrose (NP) tetrad (or null 
frame) $\{{\mbold\ell}, {\bf n}, {\bf m}^i\}$, the effect of a boost ${\mbold\ell} \mapsto e^{\lambda}{\mbold\ell}$, ${\bf n} 
\mapsto e^{-\lambda}{\bf n}$ allows $T$ to be decomposed according to 
its boost weight, 
\begin{equation} 
T=\sum_b (T)_{b}, \label{bwdecomp} 
\end{equation} 
where $(T)_{b}$ denotes the boost weight $b$ components of $T$. Recall that the boost weight $b$ components are defined as those components, $T_{ab...d}$, of $T$ that transform according to 
\[ T_{ab...d}\mapsto e^{b\lambda}T_{ab...d},\] 
under the aforementioned boost. 
 
An 
algebraic classification of tensors $T$ has been developed 
\cite{class} which is based on the existence of certain 
normal forms of (\ref{bwdecomp}) through   successive application 
of null rotations and spin-boosts.  In the special case where $T$ 
is the Weyl tensor in four dimensions, this classification reduces 
to the well-known Petrov classification. However, the boost weight 
decomposition can be used in the classification of the Weyl tensor $C$ 
in arbitrary dimensions. For the Weyl tensor we have in general,
\begin{equation} 
C=(C)_{+2}+(C)_{+1}+(C)_{0}+(C)_{-1}+(C)_{-2}, 
\end{equation} 
\noindent in every null frame.  A Weyl tensor is algebraically special if there 
exists a frame in which certain boost weight components can be transformed to zero.

\section{Curvature operators and  curvature projectors} 
A curvature operator, ${\sf T}$, is a tensor considered as a (pointwise) linear operator 
\[ {\sf T}:~V\mapsto V, \] 
for some vector space, $V$, constructed from the Riemann tensor, its covariant derivatives, and the curvature invariants. 
 
The archetypical example of a curvature operator is obtained by 
raising one index of the Ricci tensor.  The Ricci operator is 
consequently a mapping of the tangent space $T_p\mathcal{M}$ into 
itself: 
\[ {\sf R}\equiv(R^{\mu}_{~\nu}):~T_p\mathcal{M}\mapsto T_p\mathcal{M}. \] 
Another example of a curvature operator is the Weyl tensor, considered as an operator, ${\sf C}\equiv (C^{\alpha\beta}_{\phantom{\alpha\beta}\mu\nu}$), mapping bivectors onto bivectors. 
 
For a curvature operator, ${\sf T}$, consider an eigenvector ${\sf 
v}$ with eigenvalue $\lambda$; i.e., ${\sf T}{\sf v}=\lambda{\sf 
v}$. If $d=\mathrm{dim}(V)$ and $n$ is the dimension of the 
spacetime, then the eigenvalues of ${\sf T}$ are $GL(d)$ invariant. Since 
the Lorentz transformations, $O(1,n-1)$, will act via a 
representation $\Gamma\subset GL(d)$ on ${\sf T}$, 
 \emph{the eigenvalue of a curvature operator is an $O(1,n-1)$-invariant curvature scalar}. 
Therefore, curvature operators naturally provide us with a set of 
curvature invariants (not necessarily polynomial invariants) 
corresponding to the set of distinct eigenvalues: $\{\lambda_A 
\}$. Furthermore, the set of eigenvalues are uniquely determined 
by the polynomial invariants of ${\sf T}$ via its characteristic 
equation. The characteristic equation, when solved, gives us the 
set of eigenvalues, and hence these are consequently determined by 
the invariants.

We can now define a number of associated curvature operators. For 
example,  for an eigenvector ${\sf v}_A$ so that ${\sf T}{\sf 
v}_A=\lambda_A{\sf v}_{A}$, we can construct the annihilator 
operator: 
\[ {\sf P}_A\equiv ({\sf T}-\lambda_{A}{\sf 1}).\] 
Considering the Jordan block form of ${\sf T}$, the eigenvalue ${\lambda_A}$ corresponds to a set of Jordan blocks. These blocks are of the form: 
\[ {\sf B}_A=\begin{bmatrix} 
\lambda_A & 0 & 0& \cdots & 0 \\ 
1 & \lambda_A & 0& \ddots  & \vdots \\ 
0      & 1 &\lambda_A& \ddots & 0 \\ 
\vdots & \ddots     &\ddots& \ddots & 0 \\ 
0    &     \hdots &0   &  1      & \lambda_A 
\end{bmatrix}.\] 
There might be several such blocks corresponding to an eigenvalue 
$\lambda_A$; however, they are all such that $({\sf 
B}_A-\lambda_A{\sf 1})$ is nilpotent and hence there exists an 
$n_{A}\in \mathbb{N}$ such that  ${\sf P}_A^{n_A}$ annihilates the 
whole vector space associated with the eigenvalue $\lambda_A$. 
 
This implies that we can define a set of operators $\widetilde{\bot}_A$ with eigenvalues $0$ or $1$ by considering the products 
\[ \prod_{B\neq A}{\sf P}^{n_B}_B=\Lambda_A\widetilde{\bot}_A,\] 
where $\Lambda_A=\prod_{B\neq A}(\lambda_A-\lambda_B)^{n_B}\neq 0$ 
(as long as $\lambda_B\neq \lambda_A$ for all $B$). Furthermore, we can now define\footnote{This corrects a mistake in the algorithm given in \cite{inv} (thanks to Lode Wylleman for pointing this out).}
\[ \bot_A\equiv {\sf 1}-\left({\sf 1}-\widetilde{\bot}_A\right)^{n_A}  \] 
 where $\bot_A$ 
is a \emph{curvature projector}. The set of all such curvature 
projectors obeys: 
\beq {\sf 1}=\bot_1+\bot_2+\cdots+\bot_A+\cdots, 
\quad \bot_A\bot_B=\delta_{AB}\bot_A. 
\eeq We can use these 
curvature projectors to decompose the operator ${\sf T}$: \beq 
{\sf T}={\sf N}+\sum_A\lambda_A\bot_A. \label{decomp} \eeq The 
operator ${\sf N}$ therefore contains  all the information not 
encapsulated in the eigenvalues $\lambda_A$. From the Jordan form 
we can see that ${\sf N}$ is nilpotent; i.e., there exists an 
$n\in\mathbb{N}$ such that ${\sf N}^n={\sf 0}$. In particular, if 
${\sf N}\neq 0$, then ${\sf N}$ is a negative/positive 
boost weight operator which can be used to lower/raise the 
boost weight of a tensor. 
 
Considering the Weyl operator, we can show 
that (where the type refers to  Weyl type): 
\begin{itemize} 
\item{} Type I: ${\sf N}={\sf 0}$, $\lambda_A\neq 0$. 
\item{} Type D: ${\sf N}={\sf 0}$, $\lambda_A\neq 0$. 
\item{} Type II: ${\sf N}^3={\sf 0}$, $\lambda_A\neq 0$. 
\item{} Type III: ${\sf N}^3={\sf 0}$, $\lambda_A=0$. 
\item{} Type N: ${\sf N}^2={\sf 0}$, $\lambda_A=0$. 
\item{} Type O: ${\sf N}={\sf 0}$, $\lambda_A=0$. 
\end{itemize} 
 \section{Null rotations} 
\label{AppNull}
The null rotations act as follows: 
\beq
\underline{\hat{H}}^i_{~j}&=& {\hat{H}}^i_{~j}, \\
\underline{\hat{v}}^i &=&\hat{v}^i -\frac{1}{n-1}\hat{H}^i_{~k}z^k,\\
\underline{\hat{T}}^i_{~jk} &=& {\hat{T}}^i_{~jk}-2z_{[j}\hat{H}^i_{~k]}+\frac{2}{n-1}z_{l}\hat{H}^l_{~[j}\delta^i_{~k]}, \\
\underline{\bar{R}}&=& \bar{R}+4(n-1)\hat{v}^iz_i-2\hat{H}_{ij}z^iz^j, \\
\underline{\bar{R}}_{ij}&=& \bar{R}_{ij}+2\hat{T}_{(ij)k}z^k+2\delta_{ij}\hat{v}_kz^k+2(n-2)\hat{v}_{(i}z_{j)} \nonumber \\
&& +\hat{H}_{ij}|z|^2-2z_{(i}\hat{H}_{j)k}z^k, \\
\underline{A}_{ij} &=& A_{ij}+2n\hat{v}_{[i}z_{j]}-z_k\hat{T}^k_{~ij}+2z_{[i}\hat{H}_{j]k}z^k,\\
\underline{\bar{H}}_{ij\ kl}&=&{\bar{H}}_{ij\ kl}-4z_{[i}\delta_{j][k}\hat{v}_{l]}-4z_{[k}\delta_{l][i}\hat{v}_{j]}-2z_{[i}\hat{T}_{j]kl}-2z_{[k}\hat{T}_{l]ij}\nonumber \\
 && -4z_{[i}\hat{H}_{j][k}z_{l]},\\
(n-1)\underline{\check{v}}^i &=& (n-1)\check{v}^i-\tfrac 12\bar{R}z^i-\tfrac 12\bar{R}^i_{~j}z^j+\tfrac 32A^{ij}z_j+z^jz^k\hat{T}_{jki} \nonumber 
\\
&& -(2n-1)z^iz^k\hat{v}_k+\tfrac 12(n+1)\hat{v}^i|z|^2+z^i\hat{H}_{jk}z^jz^k-\tfrac12\hat{H}^i_{~j}z^j|z|^2 \\
\underline{\check{T}}^i_{~jk}&=&\check{T}^i_{~jk}+\delta^i_{~j}(\check{v}_k-\underline{\check{v}}_k)-\delta^i_{~k}(\check{v}_j-\underline{\check{v}}_j) -z^iA_{jk}+A^i_{~[j}z_{k]}+\mathcal{T}^i_{~jk}(z)\nonumber \\
&& +z_{[j}\bar{R}^i_{~k]}+z^iz_n\hat{T}^n_{~jk}-\tfrac 12|z|^2\hat{T}^i_{~jk}+2z_{[j}\hat{T}_{k]in}z^n+2nz^iz_{[j}\hat{v}_{k]}\nonumber\\
&&-|z|^2\delta^i_{~[j}\hat{v}_{k]} -2\delta^i_{~[j}z_{k]}\hat{v}_{n}z^n+2z^iz_n\hat{H}^n_{~[j}z_{k]}-|z|^2\hat{H}^i_{~[j}z_{k]} \\
\underline{\check{H}}_{ij}&=& \check{H}_{ij}+2n\check{v}_{(i}z_{j)}-2\delta_{ij}\check{v}^kz_{k}-2\check{T}_{(ij)k}z^k-z_{(i}\bar{R}_{j)k}z^k+\tfrac 12|z|^2\bar{R}_{ij}\nonumber \\
&& -\tfrac 12 \bar{R}z_iz_j+3z_{(i}A_{j)k}z^k+\mathcal{H}_{ij}(z,z)+n|z|^2z_{(i}\hat{v}_{j)}-2nz_iz_j\hat{v}_kz^k\nonumber \\&& +\delta_{ij}|z|^2\hat{v}_kz^k
+ 2z^kz^l\hat{T}_{kl(i}z_{j)}+\tfrac 12|z|^2\hat{T}_{(ij)k}z^k\nonumber \\
&& +z_iz_j\hat{H}_{kl}z^kz^l-|z|^2z^kz_{(i}\hat{H}_{j)k}+\tfrac 14|z|^4\hat{H}_{ij}
\eeq
where
\beq
\mathcal{T}_{ijk}(z)&=&z^l\bar{C}_{lijk}+\frac{2}{n-2}\left(z_{[j}\bar{R}_{k]i}-\delta_{i[j}\bar{R}_{k]l}z^l\right)\nonumber \\
&& +\frac{2\bar{R}}{(n-1)(n-2)}\delta_{i[j}z_{k]}. \\
\mathcal{H}_{ij}(z,z)&=&z^kz^l\bar{C}_{kilj}+\frac{1}{n-2}\left(|z|^2\bar{R}_{ij}-2z_{(i}\bar{R}_{j)k}z^k+\delta_{ij}\bar{R}_{kl}z^kz^l\right)\nonumber \\
&& +\frac{\bar{R}}{(n-1)(n-2)}\left(z_iz_j-|z|^2\delta_{ij}\right).
\eeq
\section{Weyl operator in orthonormal frame}
\label{AppOrtho}\label{AppOF}
Another commonly used frame is the orthonormal frame.  Let us also give the Weyl operator in such a frame; this is particularly useful for proving some of the theorems of section \ref{sect:symmetry}. Using the orthonormal frame $\{ {\mbold\omega}^t,{\mbold\omega}^x,{\bf m}^i\}$, we get the bivector basis:
\[ {\mbold\omega}^t\wedge{\bf m}^i, \quad  {\mbold\omega}^t\wedge{\mbold\omega}^x,\quad  {\bf m}^i\wedge{\bf m}^j,\quad   {\mbold\omega}^x\wedge{\bf m}^j,\] 
or for short: $[ti]$, $[tx]$, $[ij]$, $[xi]$. The  metric, $\eta_{MN}$, in bivector space is ( $m=n(n-1)/2$): 

\[ (\eta_{MN})=\frac 12 \begin{bmatrix} 
-{\sf 1}_n & 0 & 0 & 0 \\
0 & -1& 0 & 0\\
0 & 0 & {\sf 1}_m & 0\\
0 & 0 & 0 & {\sf 1}_n
\end{bmatrix},
\]
where ${\sf 1}_n$, and ${\sf 1}_m$ are the unit matrices of size $n\times n$ and $m\times m$, respectively. 
With these assumptions, the operator ${\sf C}$ can be written on the following $(n+1+m+n)$-block form (using the same definitions as eq.(\ref{WeylOperator}))
\beq
{\sf C}=\begin{bmatrix}
\tfrac 12\left(M+M^t-\check{H}-\hat{H}\right) & \tfrac{1}{\sqrt{2}}\left(\hat{K}+\check{K}\right) & -\tfrac{1}{\sqrt{2}}\left(\hat{L}-\check{L}\right) & \tfrac 12\left(M^t-M+\check{H}-\hat{H}\right) \\
\tfrac{1}{\sqrt{2}}\left(\hat{K}^t+\check{K}^t\right) & -\Phi & A^t & \tfrac{1}{\sqrt{2}}\left(\hat{K}^t-\check{K}^t\right) \\
\tfrac{1}{\sqrt{2}}\left(\hat{L}^t-\check{L}^t\right) & -A & \bar{H} & \tfrac{1}{\sqrt{2}}\left(\hat{L}^t+\check{L}^t\right) \\
\tfrac 12\left(M^t-M-\check{H}+\hat{H}\right)  & -\tfrac{1}{\sqrt{2}}\left(\hat{K}-\check{K}\right) & \tfrac{1}{\sqrt{2}}\left(\hat{L}+\check{L}\right) & \tfrac 12\left(M+M^t+\check{H}+\hat{H}\right) 
\end{bmatrix}
\eeq

\end{document}